\pdfoutput=1

\documentclass{llncs}
\usepackage{amsmath,amssymb}
\usepackage{tikz-cd}
\usepackage{array}
\usepackage{tikz}
\usepackage{pgfplots}
\usepackage{clrscode3e}
\usepackage{hyperref}
\hypersetup{
    colorlinks,
    linkcolor={blue!70!black},
    citecolor={blue!70!black},
    urlcolor={blue!70!black}
}
\usepackage{todonotes}
\usepackage{thm-restate}
\usepackage{eucal}
\usepackage{calc}
\usepackage{adjustbox}
\usepackage{wrapfig}
\usepackage{cutwin}
\usepackage{microtype}
\usepackage{multirow,bigdelim}
\usepackage{mathtools}
\usepackage{subfig}
\usepackage{mathrsfs}
\usepackage{hyperref}
\usepackage{chngcntr}
\usepackage{apptools}

\AtAppendix{\counterwithin{theorem}{section}}
\AtAppendix{\counterwithin{theorem}{section}}

\usetikzlibrary{cd,shapes,automata}

\renewcommand{\gets}{\leftarrow}

\newcommand{\Ps}{\mathcal{P}}
\newcommand{\Psc}{\Ps_c}

\newcommand{\D}{\mathcal{D}}
\newcommand{\R}{\mathbb{R}}
\newcommand{\N}{\mathbb{N}}

\newcommand{\eword}{\varepsilon}

\newcommand{\minf}{\mathsf{min}}
\newcommand{\maxf}{\mathsf{max}}

\newcommand{\lang}{\mathcal{L}}

\newcommand{\supp}{\mathsf{supp}}
\newcommand{\Conv}{\mathsf{Conv}}

\newcommand*{\circled}[1]{{\tikz[baseline=(X.base)]\node(X)[draw,shape=circle,inner sep=0]{\text{\scriptsize\strut$#1$}};}}

\newenvironment{automaton}{\begin{tikzpicture}[node distance=1.5cm,on grid,auto,initial text={},state/.append style={inner sep=0pt,outer sep=0pt,minimum size=3.5ex}]}{\end{tikzpicture}}

\title{Convex Language Semantics for \\ Nondeterministic Probabilistic
Automata\thanks{%
  This work was partially supported by ERC starting grant ProFoundNet (679127),
  ERC consolidator grant AVS-ISS (648701), a Leverhulme Prize (PLP-2016-129),
  and an NSF grant (1637532).}}
\author{Gerco van Heerdt\inst{1} \and Justin Hsu\inst{2} \and Jo\"el Ouaknine\inst{3} \and Alexandra Silva\inst{1}}
\institute{University College London \and Cornell University \and MPI-SWS and Oxford University}

\begin{document}

\maketitle

\begin{abstract}
  We explore language semantics for automata combining probabilistic and
  nondeterministic behavior.  We first show that there are precisely two natural
  semantics for probabilistic automata with nondeterminism.  For both choices,
  we show that these automata are strictly more expressive than deterministic
  probabilistic automata, and we prove that the problem of checking language
  equivalence is undecidable by reduction from the threshold problem.  However,
  we provide a discounted metric that can be computed to arbitrarily high
  precision.
\end{abstract}

\pagestyle{plain}

\section{Introduction}

Probabilistic automata are fundamental models of randomized computation. They
have been used in the study of such topics as the semantics and correctness of probabilistic
programming languages~\cite{DBLP:conf/focs/Kozen79,DBLP:conf/tacas/LegayMOW08}, randomized
algorithms~\cite{DBLP:conf/tacas/HintonKNP06,DBLP:conf/cai/Baier13,Swaminathan2012}, and machine learning~\cite{DBLP:journals/ml/BalleCG14,Ron1996}.
Removing randomness but adding nondeterminism, nondeterministic automata are
established tools for describing concurrent and distributed
systems~\cite{DBLP:journals/tcs/SassoneNW96}.

Interest in systems that exhibit
both random and nondeterministic behavior goes back to Rabin's randomized
techniques to increase the efficiency of distributed algorithms in the
1970s and 1980s~\cite{rabin76,DBLP:journals/jcss/Rabin82}.  This line of
research yielded several automata models supporting both nondeterministic and
probabilistic choice~\cite{segala-thesis,bernardo,DBLP:conf/fmco/HermannsK09}.
Many formal techniques and tools were developed for these models, and they have
been successfully used in verification
tasks~\cite{DBLP:journals/ife/Henzinger13,DBLP:conf/cav/KwiatkowskaNP11,DBLP:conf/fmco/HermannsK09},
but there are many ways of combining nondeterminism and randomization, and there
remains plenty of room for further investigation.

In this paper we study nondeterministic probabilistic automata (NPAs) and
propose a novel probabilistic language semantics. NPAs are similar to Segala
systems~\cite{segala-thesis} in that transitions can make combined
nondeterministic and probabilistic choices, but NPAs also have an output
weight in $[0, 1]$ for each state, reminiscent of
observations in Markov Decision Processes. This enables us to define the
expected weight associated with a word in a similar way to what one would do for
standard nondeterministic automata---the output of an NPA on an input word can
be computed in a \emph{deterministic} version of the automaton, using a careful
choice of algebraic structure for the state space.

Equivalence in our semantics is language equivalence (also known as trace
equivalence), which is coarser than probabilistic
bisimulation~\cite{valeria,DBLP:conf/concur/Bonchi0S17,yuxin,DBLP:conf/concur/HermannsKK14},
which distinguishes systems with different branching structure even if the total
weight assigned to a word is the same. This generalizes the classical difference
between branching and linear semantics~\cite{moshe} to the probabilistic
setting, with different target applications calling for different semantics.

After reviewing mathematical preliminaries in Section~\ref{sec:prelim}, we
introduce the NPA model and explore its semantics in Section~\ref{sec:npa}.  We
show that there are precisely two natural ways to define the language semantics
of such systems---by either taking the maximum or the minimum of the weights
associated with the different paths labeled by an input word.  The proof of this
fact relies on an abstract view on these automata generating probabilistic
languages with algebraic structure. Specifically, probabilistic languages have
the structure of a \emph{convex algebra}, analogous to the join-semilattice
structure of standard languages. These features can abstractly be seen as
so-called \emph{Eilenberg-Moore algebras} for a monad---the distribution and the
powerset monads, respectively---which can support new semantics and proof
techniques (see, e.g,~\cite{DBLP:conf/concur/Bonchi0S17,bonchi-pous}). 

Then, we compare NPAs with standard, deterministic probabilistic automata
(DPAs) (sometimes called \emph{reactive systems} in the literature) in
Section~\ref{sec:expressive}. Our semantics ensures that NPAs recover DPAs in
the special case when there is no nondeterministic choice. More interestingly,
we show that there are weighted languages accepted by NPAs that are not
accepted by any DPA. We use the theory of linear recurrence sequences to give a
separation even for weighted languages over a unary alphabet.

In Section~\ref{sec:equiv}, we turn to equivalence.  We prove that language
equivalence of NPAs is undecidable by reduction from so-called \emph{threshold
problems}, which were shown to be undecidable by Blondel and
Canterini~\cite{blondel2003}. The hard instances encoding the threshold problem are
equivalences between probabilistic automata over a two-letter alphabet. Thus,
the theorem immediately implies that equivalence of NPAs is undecidable when the
alphabet size is at least two. The situation for automata over unary
alphabets is more subtle; in particular, the threshold problem over a unary
alphabet is not known to be undecidable. However, we give a reduction from the
Positivity problem on linear recurrence sequences, a problem where a decision
procedure would necessarily entail breakthroughs in open problems in number
theory~\cite{ouaknine2014positivity}. Finally, we show that despite the
undecidability result we can provide a discounted metric that can be computed to
arbitrarily high precision.

We survey related work and conclude in Section~\ref{sec:conclusion}.

\section{Preliminaries} \label{sec:prelim}

Before we present our main technical results, we review some necessary
mathematical background on convex algebra, monads, probabilistic automata, and language semantics.

\subsection{Convex Algebra}

A set $A$ is a \emph{convex algebra}, or a \emph{convex set}, if for all $n \in \mathbb{N}$ and tuples
$(p_i)_{i = 1}^n$ of numbers in $[0, 1]$ summing up to $1$ there is an
operation denoted $\sum_{i = 1}^n p_i (-)_i \colon A^n \to A$ satisfying the
following properties for $(a_1, \dots, a_n) \in A^n$:
\begin{description}
  \item[Projection.]
    If $p_j = 1$ (and hence $p_i = 0$ for all $i \neq j$), we have $\sum_{i =
    1}^n p_i a_i = a_j$.
  \item[Barycenter.]
    For any $n$ tuples $(q_{i, j})_{j = 1}^m$ in $[0, 1]$ summing up to $1$, we
    have
    \[
      \sum_{i = 1}^n p_i \left( \sum_{j = 1}^m q_{i, j} a_j \right)
      =
      \sum_{j = 1}^m \left( \sum_{i = 1}^n p_i q_{i, j} \right) a_j .
    \]
\end{description}
Informally, a convex algebra structure gives a way to take finite convex
combinations of elements in a set $A$. Given this structure, we can define
convex subsets and generate them by elements of $A$.
\begin{definition}
  A subset $S \subseteq A$ is \emph{convex} if it is closed under all convex
  combinations. (Such a set can also be seen as a convex subalgebra.) A convex
  set $S$ is \emph{generated} by a set $G \subseteq A$ if for all $s \in S$,
	there exist $n \in \mathbb{N}$, $(p_i)_{i = 1}^n$, $(g_i)_{i = 1}^n \in G^n$ such that $s =
  \sum_i p_i g_i$. When $G$ is finite, we say that $S$ is \emph{finitely
  generated}.
\end{definition}
We can also define morphisms between convex sets.
\begin{definition}
  An \emph{affine map} between two convex sets $A$ and $B$
  is a function $h \colon A \to B$ commuting with convex combinations:
  \[
    h \left( \sum_{i = 1}^n p_i a_i \right)
    = \sum_{i = 1}^n p_i h(a_i) .
  \]
\end{definition}

\subsection{Monads and their Algebras}

Our definition of language semantics will be based on the category theoretic framework of monads and their algebras.
Monads can be used to model computational side-effects such as nondeterminism and probabilistic choice.
An algebra allows us to interpret such side-effects within an object of the category.

\begin{definition}
  A \emph{monad} $(T, \eta, \mu)$ consists of an endofunctor $T$ and two natural
  transformations: a \emph{unit} $\eta \colon \mathit{Id} \Rightarrow T$ and a
  \emph{multiplication} $\mu \colon TT \Rightarrow T$, making the following
  diagrams commute.
	\begin{align*}
		\begin{tikzcd}[ampersand replacement=\&]
			T \ar{r}{\eta} \ar[equal]{rd}{} \ar{d}[swap]{T\eta} \&
				TT \ar{d}{\mu} \\
			TT \ar{r}{\mu} \&
				T
		\end{tikzcd} &
			&
			\begin{tikzcd}[ampersand replacement=\&]
				TTT \ar{r}{T\mu} \ar{d}[swap]{\mu} \&
					TT \ar{d}{\mu} \\
				TT \ar{r}{\mu} \&
					T
			\end{tikzcd}
	\end{align*}
\end{definition}
When there is no risk of confusion, we identify a monad with its endofunctor.
An example of a monad in the category of sets is the triple $(\Ps, \{-\},
\bigcup)$, where $\Ps$ denotes the finite powerset functor sending each set to
the set of its finite subsets, $\{-\}$ is the singleton operation, and $\bigcup$
is set union.

\begin{definition}
  An \emph{algebra for a monad} $(T, \eta, \mu)$ is a pair $(X, h)$ consisting
  of a carrier set $X$ and a function $h \colon TX \to X$ making the following
  diagrams commute.
	\begin{align*}
		\begin{tikzcd}[ampersand replacement=\&]
			X \ar{r}{\eta} \ar[equal]{rd}{} \&
				TX \ar{d}{h} \\
			\&
				X
		\end{tikzcd} &
			&
			\begin{tikzcd}[ampersand replacement=\&]
				TTX \ar{r}{Th} \ar{d}[swap]{\mu} \&
					TX \ar{d}{h} \\
				TX \ar{r}{h} \&
					X
			\end{tikzcd}
	\end{align*}
\end{definition}

\begin{definition}
	A homomorphism from an algebra $(X, h)$ to an algebra $(Y, k)$ for a monad $T$ is a function $f \colon X \to Y$ making the diagram below commute.
	\[
		\begin{tikzcd}
			TX \ar{r}{Tf} \ar{d}[swap]{h} &
				TY \ar{d}{k} \\
			X \ar{r}{f} &
				Y
		\end{tikzcd}
	\]
\end{definition}

The algebras for the finite powerset monad are precisely the join-semilattices with bottom, and their homomorphisms are maps that preserve finite joins.
The algebras for any monad together with their homomorphisms form a category.

\subsection{Distribution and Convex Powerset Monads}

We will work with two monads closely associated with convex sets.  In the
category of sets, the \emph{distribution monad} $(\D, \delta, m)$ maps a set $X$
to the set of distributions over $X$ with finite support. The unit $\delta
\colon X \to \D{X}$ maps $x \in X$ to the point distribution at $x$. For the
multiplication $m \colon \D\D{X} \to \D{X}$, let $d \in \D\D{X}$ be a finite
distribution with support $\{ d_1, \dots, d_n \} \subseteq \D{X}$ and define
$m(d) = \sum_{i = 1}^n p_i d_i$, where $p_i$ is the probability of producing
$d_i$ under $d$. The category of algebras for the distribution monad is
precisely the category of convex sets and affine maps---we will often convert
between these two representations implicitly.

In the category of convex sets, the \emph{finitely generated nonempty convex
powerset monad}~\cite{DBLP:conf/concur/Bonchi0S17} $(\Psc, \{ - \}, \bigcup)$
maps a convex set $A$ to the set of finitely generated nonempty convex subsets
of $A$.\footnote{%
	The monad defined in the paper cited does not have the restriction of containing only convex subsets that are finitely generated.
	However, since all the monad operations preserve this property, the restricted
  monad is also well-defined.}
The convex algebra structure on $\Psc{A}$ is given by $\sum_{i = 1}^n p_iU_i =
\{\sum_{i = 1}^n p_i u_i \mid u_i \in U_i\text{ for all $1 \le i \le n$}\}$ with
every $U_i \in \Psc{A}$.  The unit map $\{ - \} \colon A \to \Psc{A}$ maps $a
\in A$ to a singleton convex set $\{ a \}$, and the multiplication $\bigcup
\colon \Psc\Psc{A} \to \Psc{A}$ is again the union operation, which collapses
nested convex sets.

As an example, we can consider this monad on the convex algebra $[0, 1]$.
The result is a finitely generated convex set.
\begin{lemma}\label{lem:gen}
	The convex set $\Psc[0, 1]$ is generated by its elements $\{0\}$, $\{1\}$, and $[0, 1]$, i.e., $\Conv(\{\{0\}, \{1\}, [0, 1]\}) = \Psc[0, 1]$.
\end{lemma}
\begin{proof}
  The finitely generated nonempty convex subsets of $[0, 1]$ are of the form
  $[p, q]$ for $p, q \in [0, 1]$, and $[p, q] = p\{1\} + (q - p)[0, 1] + (1 -
  q)\{0\}$.
	\qed
\end{proof}

To describe automata with both nondeterministic and probabilistic
transitions, we will work with convex powersets of distributions. The functor
$\Psc\D$ taking sets $X$ to the set of finitely generated nonempty convex sets of distributions over $X$ can
be given a monad structure.

Explicitly, writing $\omega_A \colon \D\Psc{A} \to \Psc{A}$ for the (affine) convex algebra structure on $\Psc{A}$ for any convex algebra $A$, the composite monad $(\Psc\D, \hat{\delta}, \hat{m})$ is given by
\begin{align}\label{eq:composite}
	\begin{tikzcd}[ampersand replacement=\&]
		X \ar[dashed,bend left=10]{rd}{\hat{\delta}} \ar{d}[swap]{\delta} \\
		\D{X} \ar{r}{\{-\}} \&
			\Psc\D{X}
	\end{tikzcd} &
		&
		\begin{tikzcd}[ampersand replacement=\&]
			\Psc\D\Psc\D{X} \ar[dashed,bend left=10]{rd}{\hat{m}} \ar{d}[swap]{\Psc\omega} \\
			\Psc\Psc\D{X} \ar{r}{\bigcup} \&
				\Psc\D{X}
		\end{tikzcd}
\end{align}

For all convex sets $A$ and finite nonempty subsets $S \subseteq A$, we can
define the \emph{convex closure} of $S$ (sometimes called the \emph{convex
hull}) $\Conv(S) \in \Psc{A}$ by
\[
	\Conv(S) = \{\alpha(d) \mid d \in \D{A}, \supp(d) \subseteq S\} ,
\]
where $\alpha \colon \D{A} \to A$ is the convex algebra structure on $A$.
$\Conv$ is in fact a natural transformation, a fact we will use later.

\begin{lemma}\label{lem:nat}
  For all convex sets $(A, \alpha)$ and $(B, \beta)$, affine maps $f \colon A
  \to B$, and finite nonempty subsets $S \subseteq A$, $(\Psc{f} \circ \Conv)(S)
  = (\Conv \circ \Ps{f})(S)$.
\end{lemma}
\begin{proof}
	We will first show that
	\begin{equation}\label{eq:distsets}
		\{\D{f}(d) \mid d \in \D{A}, \supp(d) \subseteq S\} = \{d \in \D{B} \mid \supp(d) \subseteq \{f(a) \mid a \in S\}\}
	\end{equation}
	for all finite nonempty $S \subseteq A$.
	For the inclusion from left to right, note that for each $d \in \D{A}$ such that $\supp(d) \subseteq S$ we have $b \in \supp(\D{f}(d))$ only if there exists $a \in S$ such that $f(a) = b$.
	Thus, $\supp(\D{f}(d)) \subseteq \{f(a) \mid a \in S\}$.
	Conversely, consider $d \in \D{B}$ such that $\supp(d) \subseteq \{f(a) \mid a \in S\}$.
	We define $d' \in \D{A}$ by
	\[
		d'(a) = \frac{d(f(a))}{|\{a' \in S \mid f(a') = f(a)\}|}.
	\]
	Then
	\begin{align*}
		\D{f}(d')(b) &
			= \sum_{a \in A, f(a) = b}d'(a) &
			&
			\text{(definition of $\D{f}$)} \\
		&
			= \sum_{a \in A, f(a) = b}\frac{d(f(a))}{|\{a' \in S \mid f(a') = f(a)\}|} &
			&
			\text{(definition of $d'$)} \\
		&
			= \sum_{a \in A, f(a) = b}\frac{d(b)}{|\{a' \in S \mid f(a') = b\}|} = d(b).
	\end{align*}

	Now we have
	\begin{align*}
		&
			\phantom{{} = {}}(\Psc{f} \circ \Conv)(S) \\
		&
			= \Psc{f}(\{\alpha(d) \mid d \in \D{A}, \supp(d) \subseteq S\}) &
			&
			\text{(definition of $\Conv$)} \\
		&
			= \{f(\alpha(d)) \mid d \in \D{A}, \supp(d) \subseteq S\} &
			&
			\text{(definition of $\Psc{f}$)} \\
		&
			= \{\beta(\D{f}(d)) \mid d \in \D{A}, \supp(d) \subseteq S\} &
			&
			\text{($f$ is affine)} \\
		&
			= \{\beta(d) \mid d \in \D{B}, \supp(d) \subseteq \{f(a) \mid a \in S\}\} &
			&
			\eqref{eq:distsets} \\
		&
			= \Conv(\{f(a) \mid a \in S\}) &
			&
			\text{(definition of $\Conv$)} \\
		&
			= (\Conv \circ \Ps{f})(S) &
			&
			\text{(definition of $\Ps{f}$)}.
	\end{align*}
	\qed
\end{proof}

\subsection{Automata and Language Semantics}\label{sec:automata}

In this section we review the general language semantics for automata with
side-effects provided by a monad (see,
e.g.,~\cite{arbib1975_,SilvaBBR13,goncharov2014towards}). This categorical
framework is the foundation of our language semantics for NPA.

\begin{definition}\label{def:automaton}
  Given a monad $(T, \eta, \mu)$ in the category of sets, an output set $O$, and
  a (finite) alphabet $A$, a $T$-\emph{automaton} is defined by a tuple $(S,
  s_0, \gamma, \{ \tau_a \}_{a \in A})$, where $S$ is the set of \emph{states},
  $s_0 \in S$ is the \emph{initial state}, $\gamma \colon S \to O$ is the
  \emph{output function}, and $\tau_a \colon S \to TS$ for $a \in A$ are the
  \emph{transition functions}.
\end{definition}

This abstract formulation encompasses many standard notions of automata.  For
instance, we recover deterministic (Moore) automata by letting $T$ be the
identity monad; deterministic acceptors are a further specialization where the
output set is the set $2 = \{0, 1\}$, with $0$ modeling rejecting states and $1$
modeling accepting states. If we use the powerset monad, we recover
nondeterministic acceptors. 

Any $T$-automaton can be determinized, using a categorical generalization of the
powerset construction~\cite{SilvaBBR13}.

\begin{definition}
  Given a monad $(T, \eta, \mu)$ in the category of sets, an output set $O$ with
  a $T$-algebra structure $o \colon TO \to O$, and a (finite) alphabet $A$, a
  $T$-automaton $(S, s_0, \gamma, \{ \tau_a \}_{a \in A})$ can be determinized into
  the deterministic automaton $(TS, s_0', \gamma', \{ \tau_a' \}_{a \in A})$
  given by $s_0' = \eta(s_0) \in TS$ and
	\begin{align*}
		\gamma' \colon TS \to O &
			&
			\tau_a' \colon TS \to TS \\
		\gamma' = o \circ T\gamma &
			&
			\tau_a' = \mu \circ T\tau_a.
	\end{align*}
\end{definition}

This construction allows us to define the language semantics of any
$T$-automaton as the semantics of its determinization. More formally, we have
the following definition.

\begin{definition}\label{def:lang}
  Given a monad $(T, \eta, \mu)$ in the category of sets, an output set $O$ with
  a $T$-algebra structure $o \colon TO \to O$, and a (finite) alphabet $A$, the
  \emph{language} accepted by a $T$-automaton $\mathcal{A} = (S, s_0, \gamma, \{
  \tau_a \}_{a \in A})$ is the function $\lang_{\mathcal{A}} \colon A^* \to O$
  given by $\lang_{\mathcal{A}} = (l_{\mathcal{A}} \circ \eta)(s_0)$, where
  $l_{\mathcal{A}} \colon TS \to O^{A^*}$ is defined inductively by
	\begin{align*}
		l_{\mathcal{A}}(s)(\eword) = (o \circ T\gamma)(s) &
			&
			l_{\mathcal{A}}(s)(av) = l_{\mathcal{A}}((\mu \circ T\tau_a)(s))(v).
	\end{align*}
\end{definition}

As an example, we recover deterministic probabilistic automata (DPAs) by
taking $T$ to be the distribution monad $\D$ and letting the output set be the
interval $[0, 1]$.  That is, a DPA with finite\footnote{%
  Concrete automata considered in this paper will have a finite state space, but
the general automaton in Definition~\ref{def:automaton} does not have this
restriction. The distribution monad, for example, does not preserve finite sets
in general.}
state space $S$ has an output function of type $S \to [0, 1]$, and each of its
transition functions is of type $S \to \D{S}$.  In order to derive a semantics
for this automata, we use the usual $\D$-algebra structure $\mathbb{E} \colon
\D[0, 1] \to [0, 1]$ computing the expected weight.

More concretely, the semantics works as follows.  Let $(S, s_0, \gamma, \{
\tau_a \}_{a \in A})$ be a DPA.  At any time while reading a word, we are in a
convex combination of states $\sum_{i = 1}^n p_i s_i$ (equivalently, a
distribution over states).  The current output is given by evaluating the sum
$\sum_{i = 1}^n p_i \gamma(s_i)$.  On reading a symbol $a \in A$, we transition
to the convex combination of convex combinations $\sum_{i = 1}^n p_i
\tau_a(s_i)$, say $\sum_{i = 1}^n p_i \sum_{j = 1}^{m_i} q_{i, j} s_{i, j}$,
which is collapsed to the final convex combination $\sum_{i = 1}^n \sum_{j =
1}^{m_i} p_i q_{i, j} s_{i, j}$ (again, a distribution over states).

\begin{remark}\label{rem:initial}
  One may wonder if the automaton model would be more expressive if the initial
  state $s_0$ in an automaton $(S, s_0, \gamma, \{ \tau_a \}_{a \in A})$ would
  be an element of $TS$ rather than $S$. This is not the case, since we can
  always add a new element to $S$ that simulates $s_0$ by setting its output to
  $(o \circ T\gamma)(s_0)$ and its transition on $a \in A$ to $(\mu \circ
  T\tau_a)(s_0)$.

  For instance, DPAs allowing a distribution over states as the initial state
  can be represented by an initial state distribution $\mu$, an output vector
  $\gamma$, and transitions $\tau_a$. In typical presentations, $\mu$ and
  $\gamma$ are represented as weight vectors over states, and the $\tau_a$ are
  encoded by stochastic matrices.
\end{remark}

\section{Nondeterministic Probabilistic Automata} \label{sec:npa}

We work with an automaton model supporting probabilistic and nondeterministic
behavior, inspired by Segala~\cite{segala-thesis}. On each input letter, the
automaton can choose from a finitely generated nonempty convex set of
distributions over states. After selecting a distribution, the automaton then
makes a probabilistic transition to the following state. Each state has an
output weight in $[0, 1]$.  The following formalization is an instantiation of
Definition~\ref{def:automaton} with the monad $\Psc\D$.

\begin{definition}
	A \emph{nondeterministic probabilistic automaton} (NPA) over a (finite)
	alphabet $A$ is defined by a tuple $(S, s_0, \gamma, \{ \tau_a \}_{a \in
	A})$, where $S$ is a finite set of \emph{states}, $s_0 \in S$ is the
	\emph{initial state}, $\gamma \colon S \to [0, 1]$ is the \emph{output function},
	and $\tau_a \colon S \to \Psc\D{S}$ are the \emph{transition functions}.
\end{definition}

As an example, consider the NPA below.
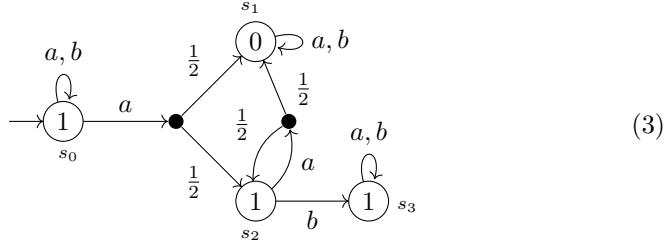
\begin{equation}\label{eq:npa}
	\begin{gathered}
		\begin{automaton}
			\node[initial,state] (q0) {$1$};
			\node[fill,circle,inner sep=0pt,minimum size=.2cm] (n1) [right of=q0] {};
			\node[state] (q1) [above right of=n1] {$0$};
			\node[state] (q2) [below right of=n1] {$1$};
			\node[fill,circle,inner sep=0pt,minimum size=.2cm] (n2) [right of=n1] {};
			\node[state] (q3) [right of=q2] {$1$};
			\path[->]
			(q0) edge [loop above] node {$a, b$} ()
			(q0) edge node {$a$} (n1)
			(q0) edge [draw=none,loop below,pos=0] node {\scriptsize $s_0$} ()
			(n1) edge node {$\frac{1}{2}$} (q1)
			(n1) edge [swap] node {$\frac{1}{2}$} (q2)
			(q1) edge [loop right] node {$a, b$} ()
			(q1) edge [draw=none,loop above,pos=0] node {\scriptsize $s_1$} ()
			(q2) edge [bend right,swap,pos=.7] node {$a$} (n2)
			(q2) edge [swap] node {$b$} (q3)
			(q2) edge [draw=none,loop below,pos=1] node {\scriptsize $s_2$} ()
			(n2) edge [swap,pos=0] node {$\frac{1}{2}$} (q1)
			(n2) edge [bend right,swap] node {$\frac{1}{2}$} (q2)
			(q3) edge [loop above] node {$a, b$} ()
			(q3) edge [draw=none,loop right,pos=1] node {\scriptsize $s_3$} ();
		\end{automaton}
	\end{gathered}
\end{equation}
States are labeled by their direct output (i.e., their weight from $\gamma$)
while outgoing edges represent transitions.  Additionally, we write the state
name next to each state.  We only indicate a set of generators of the convex
subset that a state transitions into. If one of these generators is a
distribution with nonsingleton support, then a transition into a black dot is
depicted, from which the outgoing transitions represent the distribution.  Those
edges are labeled with probabilities.

Our NPAs recognize weighted languages. The rest of the section is concerned with
formally defining this semantics, based on the general framework from
Section~\ref{sec:automata}.

\subsection{From Convex Algebra to Language Semantics}

To define language semantics for NPAs, we will use the monad structure of
$\Psc\D$.  To be able to use the semantics from Section~\ref{sec:automata}, we
need to specify a $\Psc\D$-algebra structure $o \colon \Psc\D[0, 1] \to [0, 1]$.
Moreover, our model should naturally coincide with DPAs when transitions make no
nondeterministic choices, i.e., when each transition function maps each state to
a singleton distribution over states.  Thus, we require the $\Psc\D$-algebra $o$
to extend the expected weight function $\mathbb{E}$, making the diagram below
commute.
\begin{equation}\label{eq:extend}
	\begin{tikzcd}
		\D[0, 1] \ar{d}[swap]{\{-\}} \ar{rd}{\mathbb{E}} \\
		\Psc\D[0, 1] \ar{r}{o} &
			{[0, 1]}
	\end{tikzcd}
\end{equation}

\subsection{Characterizing the Convex Algebra on $[0, 1]$}

While in principle there could be many different $\Psc\D$-algebras on $[0, 1]$
leading to different language semantics for NPAs, we show that (i) each
algebra extending the $\D$-algebra on $[0, 1]$ is fully determined by a
$\Psc$-algebra on $[0, 1]$, and (ii) there are exactly two $\Psc$-algebras on
$[0, 1]$: the map computing the minimum and the map computing the maximum.

\begin{proposition}
	Any $\Psc\D$-algebra on $[0, 1]$ extending $\mathbb{E} \colon \D[0, 1] \to [0, 1]$ is of the form $\Psc\D[0, 1] \xrightarrow{\Psc{\mathbb{E}}} \Psc[0, 1] \xrightarrow{\alpha} [0, 1]$, where $\alpha$ is a $\Psc$-algebra.
\end{proposition}
\begin{proof}
	Let $o \colon \Psc\D[0, 1] \to [0, 1]$ be a $\Psc\D$-algebra extending $\mathbb{E}$.
	We define
	\[
		\alpha = \Psc[0, 1] \xrightarrow{\Psc\delta} \Psc\D[0, 1] \xrightarrow{o} [0, 1].
	\]
	Indeed, the diagram
	
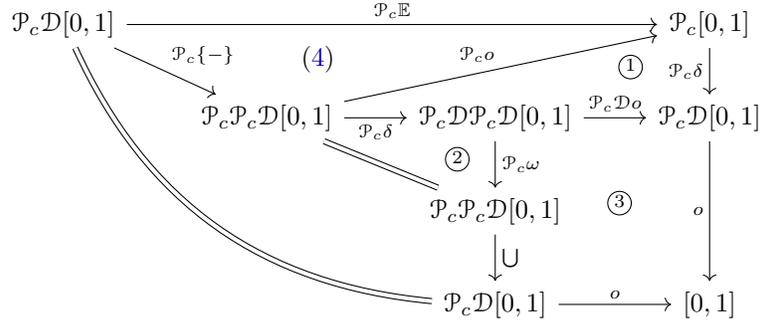
\begin{gather*}
		\begin{tikzcd}[ampersand replacement=\&]
			\Psc\D[0, 1] \ar{rrr}{\Psc\mathbb{E}} \ar{rd}{\Psc\{-\}} \ar[equal,bend right=30]{rrddd}{} \ar[bend right=10,phantom]{rrr}[pos=.35]{\eqref{eq:extend}} \&
				\&
				\&
				\Psc[0, 1] \ar{d}[swap]{\Psc\delta} \ar[bend right=80,phantom]{d}{\circled{1}} \\
			\&
				\Psc\Psc\D[0, 1] \ar{r}[swap]{\Psc\delta} \ar{rru}{\Psc{o}} \ar[equal]{rd}{} \&
				\Psc\D\Psc\D[0, 1] \ar{r}{\Psc\D{o}} \ar{d}{\Psc\omega} \ar[bend right=50,phantom]{d}[pos=.4]{\circled{2}} \ar[bend left=10,phantom]{rdd}{\circled{3}} \&
				\Psc\D[0, 1] \ar{dd}[swap]{o} \\
			\&
				\&
				\Psc\Psc\D[0, 1] \ar{d}{\bigcup} \\
			\&
				\&
				\Psc\D[0, 1] \ar{r}{o} \&
				{[0, 1]}
		\end{tikzcd} \\
		\begin{array}{l@{\quad}l@{\quad}l}
			\circled{1}\text{ naturality of $\delta$} &
				\circled{2}\text{ $\omega$ is a convex algebra} &
				\circled{3}\text{ $o$ is a $\Psc\D$-algebra}
		\end{array}
	\end{gather*}
	commutes, so it only remains to show that $\alpha$ is a $\Psc$-algebra.
	This can be seen from the commutative diagrams below.
	
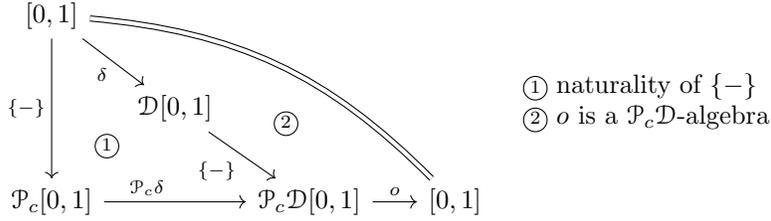
\[
		\begin{tikzcd}[column sep=.3cm]
			{[0, 1]} \ar[equal,bend left=20]{rrrrdd}{} \ar{rd}[swap]{\delta} \ar{dd}[swap]{\{-\}} \ar[phantom]{rrrrdd}[pos=.6]{\circled{2}} \\
			&
				\D[0, 1] \ar{rd}[swap]{\{-\}} \\
			\Psc[0, 1] \ar{rr}[pos=.3]{\Psc\delta} \ar[bend left=10,phantom]{ur}[pos=.65]{\circled{1}} &
				&
				\Psc\D[0, 1] \ar{rr}{o} &
				&
				{[0, 1]}
		\end{tikzcd}
		\quad
		\begin{array}{l}
			\circled{1}\text{ naturality of $\{-\}$} \\
			\circled{2}\text{ $o$ is a $\Psc\D$-algebra}
		\end{array}
	\]
	
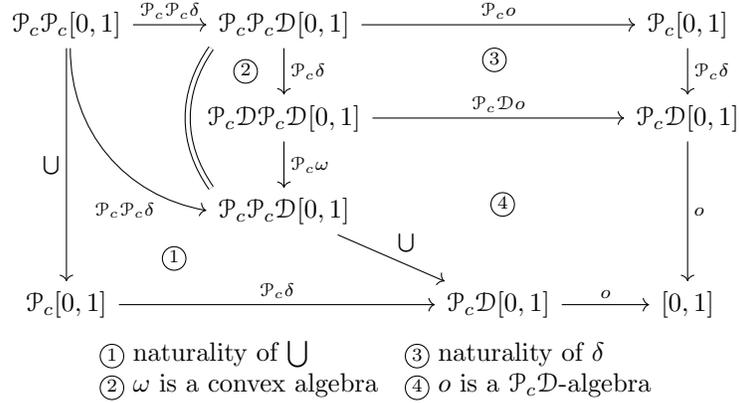
\begin{gather*}
		\begin{tikzcd}[ampersand replacement=\&]
			\Psc\Psc[0, 1] \ar{r}{\Psc\Psc\delta} \ar[bend right=40]{rdd}[swap,pos=.8]{\Psc\Psc\delta} \ar{ddd}[swap]{\bigcup} \&
				\Psc\Psc\D[0, 1] \ar{rr}{\Psc{o}} \ar{d}{\Psc\delta} \ar[equal,bend right=35,shift right=.75cm]{dd}{} \ar[bend right=50,phantom]{d}{\circled{2}} \ar[bend left=5,phantom]{rrd}{\circled{3}} \&
				\&
				\Psc[0, 1] \ar{d}{\Psc\delta} \\
			\&
				\Psc\D\Psc\D[0, 1] \ar{rr}{\Psc\D{o}} \ar{d}{\Psc\omega} \ar[bend left=5,phantom]{rrdd}{\circled{4}} \&
				\&
				\Psc\D[0, 1] \ar{dd}{o} \\
			\&
				\Psc\Psc\D[0, 1] \ar{rd}{\bigcup} \\
			\Psc[0, 1] \ar{rr}{\Psc\delta} \ar[phantom]{ur}{\circled{1}} \&
				\&
				\Psc\D[0, 1] \ar{r}{o} \&
				{[0, 1]}
		\end{tikzcd} \\
		\begin{array}{l@{\quad}l}
			\circled{1}\text{ naturality of $\bigcup$} &
				\circled{3}\text{ naturality of $\delta$} \\
			\circled{2}\text{ $\omega$ is a convex algebra} &
				\circled{4}\text{ $o$ is a $\Psc\D$-algebra}
		\end{array}
	\end{gather*}
	\qed
\end{proof}

\begin{proposition}
	The only $\Psc$-algebras on the convex set $[0, 1]$ are $\minf$ and $\maxf$.
\end{proposition}
\begin{proof}
	Let $\alpha \colon \Psc[0, 1] \to [0, 1]$ be a $\Psc$-algebra.
	Then for any $r \in [0, 1]$, $\alpha(\{r\}) = r$, and the diagram below must commute.
	\begin{equation}\label{eq:alg}
		\begin{tikzcd}
			\Psc\Psc[0, 1] \ar{r}{\Psc\alpha} \ar{d}[swap]{\bigcup} &
				\Psc[0, 1] \ar{d}{\alpha} \\
			\Psc[0, 1] \ar{r}{\alpha} &
				{[0, 1]}
		\end{tikzcd}
	\end{equation}
	Furthermore, $\alpha$ is an affine map.
	Since $\Conv(\{\{0\}, \{1\}, [0, 1]\}) = \Psc[0, 1]$ by Lemma~\ref{lem:gen}, $\alpha(\{0\}) = 0$, and $\alpha(\{1\}) = 1$, $\alpha$ is completely determined by $\alpha([0, 1])$.
	We now calculate that
	\begin{align*}
		\alpha([0, 1]) &
			= \alpha\left(\bigcup\{[0, p] \mid p \in [0, 1]\}\right) \\
		&
			= (\alpha \circ {\bigcup} \circ \Conv)(\{\{0\}, [0, 1]\}) \\
		&
			= (\alpha \circ \Psc\alpha \circ \Conv)(\{\{0\}, [0, 1]\}) &
			&
			\eqref{eq:alg} \\
		&
			= (\alpha \circ \Conv \circ \Ps\alpha)(\{\{0\}, [0, 1]\}) &
			&
			\text{(Lemma~\ref{lem:nat})} \\
		&
			= (\alpha \circ \Conv)(\{\alpha(\{0\}), \alpha([0, 1])\}) &
			&
			\text{(definition of $\Psc\alpha$)} \\
		&
			= (\alpha \circ \Conv)(\{0, \alpha([0, 1])\}) \\
		&
			= \alpha([0, \alpha([0, 1])]) \\
		&
			= \alpha(\alpha([0, 1])[0, 1] + (1 - \alpha([0, 1]))\{0\}) \\
		&
			= \alpha([0, 1]) \cdot \alpha([0, 1]) + (1 - \alpha([0, 1])) \cdot \alpha(\{0\}) &
			&
			\text{($\alpha$ is affine)} \\
		&
			= \alpha([0, 1])^2 + (1 - \alpha([0, 1])) \cdot 0 \\
		&
			= \alpha([0, 1])^2.
	\end{align*}
	Thus, we have either $\alpha([0, 1]) = 0$ or $\alpha([0, 1]) = 1$.
	Consider any finitely generated nonempty convex subset $[p, q] \subseteq [0, 1]$.
	If $\alpha([0, 1]) = 0$, then Lemma~\ref{lem:gen} gives
	\begin{align*}
		\alpha([p, q]) &
			= \alpha(p\{1\} + (q - p)[0, 1] + (1 - q)\{0\}) \\
		&
			= p \cdot \alpha(\{1\}) + (q - p) \cdot \alpha([0, 1]) + (1 - q) \cdot \alpha(\{0\}) \\
		&
			= p \cdot 1 + (q - p) \cdot 0 + (1 - q) \cdot 0 = p = \minf([p, q]);
	\end{align*}
	if $\alpha([0, 1]) = 1$, then
	\begin{align*}
		\alpha([p, q]) &
			= \alpha(p\{1\} + (q - p)[0, 1] + (1 - q)\{0\}) \\
		&
			= p \cdot \alpha(\{1\}) + (q - p) \cdot \alpha([0, 1]) + (1 - q) \cdot \alpha(\{0\}) \\
		&
			= p \cdot 1 + (q - p) \cdot 1 + (1 - q) \cdot 0 = q = \maxf([p, q]).
	\end{align*}
  We now show that $\minf$ is an algebra; the case for $\maxf$ is analogous.  We have
	\begin{align*}
		\minf\left(\sum_{i = 1}^n r_i[p_i, q_i]\right) &
			= \minf\left(\left[\sum_{i = 1}^n r_i \cdot p_i, \sum_{i = 1}^n r_i \cdot q_i\right]\right) \\
		&
			= \sum_{i = 1}^n r_i \cdot p_i \\
		&
			= \sum_{i = 1}^n r_i \cdot \minf([p_i, q_i]) ,
	\end{align*}
	so $\minf$ is an affine map.
	Furthermore, clearly $\minf(\{r\}) = r$ for all $r \in [0, 1]$, and for all $S \in \Psc\Psc[0, 1]$,
	\[
		\minf\left(\bigcup S\right) = \minf(\{\minf(T) \mid T \in S\}) = (\minf \circ \Psc\minf)(S).
	\]
	\qed
\end{proof}

\begin{corollary}\label{cor:minmax}
	The only $\Psc\D$-algebras on $[0, 1]$ extending $\mathbb{E}$ are $\Psc\D[0, 1] \xrightarrow{\Psc{\mathbb{E}}} \Psc[0, 1] \xrightarrow{\minf} [0, 1]$ and $\Psc\D[0, 1] \xrightarrow{\Psc{\mathbb{E}}} \Psc[0, 1] \xrightarrow{\maxf} [0, 1]$.
\end{corollary}

Consider again the NPA \eqref{eq:npa}.  Since we can always choose to remain in
the initial state, the $\maxf$ semantics assigns 1 to each word for this
automaton.  The $\minf$ semantics is more interesting.  Consider reading the
word $aa$.  On the first $a$, we transition from $s_0$ to $\Conv\{s_0,
\frac{1}{2}s_1 + \frac{1}{2}s_2\} \in \Psc\D{S}$.  Reading the second $a$ gives
\[
	\Conv\left\{\Conv\left\{s_0, \tfrac{1}{2}s_1 + \tfrac{1}{2}s_2\right\}, \tfrac{1}{2}\{s_1\} + \tfrac{1}{2}\left\{\tfrac{1}{2}s_1 + \tfrac{1}{2}s_2\right\}\right\} \in \Psc\D\Psc\D{S}.
\]
Now we first apply $\Psc\omega$ to eliminate the outer distribution, arriving at
\[
	\Conv\left\{\Conv\left\{s_0, \tfrac{1}{2}s_1 + \tfrac{1}{2}s_2\right\}, \left\{\tfrac{3}{4}s_1 + \tfrac{1}{4}s_2\right\}\right\} \in \Psc\Psc\D{S}.
\]
Taking the union yields
\[
	\Conv\left\{s_0, \tfrac{1}{2}s_1 + \tfrac{1}{2}s_2, \tfrac{3}{4}s_1 + \tfrac{1}{4}s_2\right\} \in \Psc\D{S},
\]
which leads to the convex subset of distributions over outputs
\[
	\Conv\left\{1, \tfrac{1}{2} \cdot 0 + \tfrac{1}{2}\cdot 1, \tfrac{3}{4} \cdot 0 + \tfrac{1}{4} \cdot 1\right\} \in \Psc\D[0, 1].
\]
Calculating the expected weights gives $\Conv\{1, \frac{1}{2}, \frac{1}{4}\} \in \Psc[0, 1]$, which has a minimum of $\frac{1}{4}$.
One can show that on reading any word $u \in A^*$ the automaton outputs $2^{-n}$, where $n$ is the length of the longest sequence of $a$'s occurring in $u$.

The semantics coming from $\maxf$ and $\minf$ are highly symmetrical; in a
sense, they are two representations of the same semantics.\footnote{%
  The $\maxf$ semantics is perhaps preferable since it recovers standard
  nondeterministic finite automata when there is no probabilistic choice and the output weights are in $\{0, 1\}$, but
this is a minor point.}
Technically, we establish the following relation between the two
semantics---this will be useful to avoid repeating proofs twice for each
property.

\begin{proposition}\label{prop:relation}
	Consider an NPA $\mathcal{A} = (S, s_0, \gamma, \{\tau_a\}_{a \in A})$ under the $\minf$ semantics.
	Define $\gamma' \colon S \to [0, 1]$ by $\gamma'(s) = 1 - \gamma(s)$, and consider the NPA $\mathcal{A}' = (S, s_0, \gamma', \{\tau_a\}_{a \in A})$ under the $\maxf$ semantics.
	Then $\lang_{\mathcal{A}'}(u) = 1 - \lang_{\mathcal{A}}(u)$ for all $u \in A^*$.
\end{proposition}
\begin{proof}
	We prove the stronger property that for all $x \in \Psc\D{S}$ and $u \in A^*$, $l_{\mathcal{A}'}(x)(u) = 1 - l_{\mathcal{A}}(x)(u)$ by induction on $u$.
	This is sufficient because $\mathcal{A}$ and $\mathcal{A}'$ have the same initial state.
	We have
	\begin{align*}
		&
			\phantom{{} = {}}l_{\mathcal{A}'}(x)(\eword) \\
		&
			= (\maxf \circ \Psc\mathbb{E} \circ \Psc\D\gamma')(x) &
			&
			\text{(Definition~\ref{def:lang})} \\
		&
			= (\maxf \circ \Psc\mathbb{E})\left(\left\{\lambda p.\,\sum_{s \in S, \gamma'(s) = p} d(s) \mathrel{}\middle|\mathrel{} d \in x\right\}\right) &
			&
			\text{(definition of $\Psc\D\gamma'$)} \\
		&
			= \maxf\left(\left\{\sum_{p \in [0, 1]} p \sum_{s \in S, \gamma'(s) = p} d(s) \mathrel{}\middle|\mathrel{} d \in x\right\}\right) &
			&
			\text{(definition of $\Psc\mathbb{E}$)} \\
		&
			= \maxf\left(\left\{\sum_{p \in [0, 1]} p \sum_{s \in S, \gamma(s) = 1 - p} d(s) \mathrel{}\middle|\mathrel{} d \in x\right\}\right) &
			&
			\text{(definition of $\gamma'$)} \\
		&
			= \maxf\left(\left\{\sum_{p \in [0, 1]} (1 - p) \sum_{s \in S, \gamma(s) = p} d(s) \mathrel{}\middle|\mathrel{} d \in x\right\}\right) \\
		&
			= \maxf\left(\left\{\sum_{p \in [0, 1]} (1 - p) \cdot \D\gamma(d)(p) \mathrel{}\middle|\mathrel{} d \in x\right\}\right) \\
		&
			= \maxf\left(\left\{1 - \sum_{p \in [0, 1]} p \cdot \D\gamma(d)(p) \mathrel{}\middle|\mathrel{} d \in x\right\}\right) \\
		&
			= 1 - \minf\left(\left\{\sum_{p \in [0, 1]} p \cdot \D\gamma(d)(p) \mathrel{}\middle|\mathrel{} d \in x\right\}\right) \\
		&
			= 1 - (\minf \circ \Psc\mathbb{E})(\{\D\gamma(d) \mid d \in x\}) &
			&
			\text{(definition of $\Psc\mathbb{E}$)} \\
		&
			= 1 - (\minf \circ \Psc\mathbb{E} \circ \Psc\D\gamma)(x) &
			&
			\text{(definition of $\Psc\D\gamma'$)} \\
		&
			= 1 - l_{\mathcal{A}}(x)(\eword) &
			&
			\text{(Definition~\ref{def:lang})}.
	\end{align*}
	Furthermore,
	\begin{align*}
		l_{\mathcal{A}'}(x)(av) &
			= l_{\mathcal{A}'}\left(\left({\bigcup} \circ \Psc\omega \circ \Psc\D\tau_a\right)(x)\right)(v) &
			&
			\text{(Definition~\ref{def:lang})} \\
		&
			= 1 - l_{\mathcal{A}}\left(\left({\bigcup} \circ \Psc\omega \circ \Psc\D\tau_a\right)(x)\right)(v) &
			&
			\text{(induction hypothesis)} \\
		&
			= 1 - l_{\mathcal{A}}(x)(av) &
			&
			\text{(Definition~\ref{def:lang})}.
	\end{align*}
	This concludes the proof.
	\qed
\end{proof}

\section{Expressive Power of NPAs} \label{sec:expressive}

Our convex language semantics for NPAs coincides with the standard semantics for
DPAs when all convex sets in the transition functions are singleton sets. In
this section, we show that NPAs are in fact strictly more expressive than DPAs.
We give two results. First, we exhibit a concrete language over a binary
alphabet that is recognizable by a NPA, but not recognizable by any DPA.  This
argument uses elementary facts about the Hankel matrix, and actually shows that
NPAs are strictly more expressive than weighted finite automata (WFAs).

Next, we separate NPAs and DPAs over a unary alphabet.  This argument is
substantially more technical, relying on deeper results from number theory about
linear recurrence sequences.

\subsection{Separating NPAs and DPAs: Binary Alphabet}

Consider the language $\lang_a \colon \{a, b\}^* \to [0, 1]$ by $\lang_a(u) =
2^{-n}$, where $n$ is the length of the longest sequence of $a$'s occurring in
$u$.
Recall that this language is accepted by the NPA \eqref{eq:npa} using the $\minf$ algebra.

\begin{theorem}
	NPAs are more expressive than DPAs.
	Specifically, there is no DPA, or even WFA, accepting $\lang_a$.
\end{theorem}
\begin{proof}
	Assume there exists a WFA accepting $\lang_a$, and let $l(u)$ for $u \in \{a, b\}^*$ be the language of the linear combination of states reached after reading the word $u$.
	We will show that the languages $l(a^nb)$ for $n \in \N$ are linearly independent.
	Since the function that assigns to each linear combination of states its accepted language is a linear map, this implies that the set of linear combinations of states of the WFA is a vector space of infinite dimension, and hence the WFA cannot exist.

	The proof is by induction on a natural number $m$.
	Assume that for all natural numbers $i \le m$ the languages $l(a^ib)$ are linearly independent.
	For all $i \le m$ we have $l(a^ib)(a^m) = 2^{-m}$ and $l(a^ib)(a^{m + 1}) = 2^{-m - 1}$; however, $l(a^{m + 1}b)(a^m) = l(a^{m + 1}b)(a^{m + 1}) = 2^{-m - 1}$.
	If $l(a^{m + 1}b)$ is a linear combination of the languages $l(a^ib)$ for $i \le m$, then there are constants $c_1, \ldots, c_m \in \R$ such that in particular
  \[
    (c_1 + \cdots + c_m)2^{-m} = 2^{-m - 1}
    \qquad\text{and}\qquad
    (c_1 + \cdots + c_m)2^{-m - 1} = 2^{-m - 1} .
  \]
  These equations cannot be satisfied.
  Therefore, for all natural numbers $i \le m + 1$ the languages $l(a^ib)$ are linearly independent.
  We conclude by induction that for all $m \in \N$ the languages $l(a^ib)$ for $i \le m$ are linearly independent, which implies that all languages $l(a^nb)$ for $n \in \N$ are linearly independent.
  \qed
\end{proof}

A similar argument works for NPAs under the $\maxf$ algebra semantics---one can
easily repeat the argument in the above theorem for the language accepted by the
NPA resulting from applying Proposition~\ref{prop:relation} to the NPA
\eqref{eq:npa}.

\subsection{Separating NPAs and DPAs: Unary Alphabet}

We now turn to the unary case. A weighted language over a unary alphabet can be
represented by a sequence $\langle u_i \rangle = u_0, u_1, \dots$ of real
numbers. We will give such a language that is recognizable by a NPA but not
recognizable by any WFA (and in particular, any DPA) using
results on \emph{linear recurrence sequences}, an established tool for studying
unary weighted languages.

We begin with some mathematical preliminaries. A sequence of real numbers
$\langle u_i \rangle$ is a \emph{linear recurrence sequence} (LRS) if for some
integer $k \in \mathbb{N}$ (the \emph{order}), constants $u_0, \dots, u_{k - 1}
\in \mathbb{R}$ (the \emph{initial conditions}), and coefficients $b_0, \dots,
b_{k - 1} \in \mathbb{R}$, we have
\[
  u_{n + k} = b_{k - 1} u_{n - 1} + \cdots + b_0 u_n
\]
for every $n \in \mathbb{N}$. A well-known example of an LRS is the
\emph{Fibonacci sequence}, an order-$2$ LRS satisfying the recurrence $f_{n + 2}
= f_{n + 1} + f_n$. Another example of an LRS is any constant sequence, i.e.,
$\langle u_i \rangle$ with $u_i = c$ for all $i$.

Linear recurrence sequences are closed under linear combinations: for any two
LRS $\langle u_i \rangle, \langle v_i \rangle$ and constants $\alpha, \beta \in
\mathbb{R}$, the sequence $\langle \alpha u_i + \beta v_i \rangle$ is again an
LRS (possibly of larger order). We will
use one important theorem about LRSs. See the monograph by Everest et
al.~\cite{DBLP:books/daglib/0035746} for details.

\begin{theorem}[Skolem-Mahler-Lech] \label{thm:sml}
	If $\langle u_i \rangle$ is an LRS, then its \emph{zero set} $\{i \in \N \mid u_i = 0\}$ is the union of a finite set along
  with finitely many arithmetic progressions (i.e., sets of the form $\{ p + k n
  \mid n \in \mathbb{N} \}$ with $k \neq 0$).
\end{theorem}

This is a celebrated result in number theory and not at all easy to prove. To
make the connection to probabilistic and weighted automata, we will use two
results. The first proposition follows from the Cayley-Hamilton Theorem.

\begin{proposition}[see, e.g.,~\cite{ouaknine2014positivity}] \label{prop:pa-to-lrs}
  Let $\lang$ be a weighted unary language recognizable by a weighted
  automaton $W$. Then the sequence of weights $\langle u_i \rangle$ with $u_i =
  \lang(a^i)$ is an LRS, where the order is at most the number of states in $W$.
\end{proposition}

While not every LRS can be recognized by a DPA, it is known that DPAs can
recognize a weighted language encoding the sign of a given LRS.

\begin{theorem}[{Akshay, et al.~\cite[Theorem 3, Corollary 4]{akshay2015reachability}}] \label{thm:lrs-to-pa}
  Given any LRS $\langle u_i \rangle$, there exists a stochastic matrix $M$ such
  that
  \[
    u_n \geq 0 \iff u^T M^n v \geq 1/4
  \]
  for all $n$, where $u = (1, 0, \dots, 0)$ and $v = (0, 1, 0, \dots, 0)$.
  Equality holds on the left if and only if it holds on the right. The language
  $\lang(a^n) = u^T M^n v$ is recognizable by a DPA with input vector
  $u$, output vector $v$, and transition matrix $M$ (Remark~\ref{rem:initial}).
  If the LRS is rational, $M$ can be taken to be rational as well.
\end{theorem}

We are now ready to separate NPAs and WFAs over a unary alphabet.

\begin{theorem}
	There is a language over a unary alphabet that is recognizable by an NPA but not by any WFA (and in particular any DPA).
\end{theorem}
\begin{proof}
  We will work in the complex numbers $\mathbb{C}$, with $i$ being the positive
  square root of $-1$ as usual. Let $a, b \in \mathbb{Q}$ be nonzero such that
  $z \triangleq a + bi$ is on the unit circle in $\mathbb{C}$, for instance $a =
  3/5, b = 4/5$ so that $|a + bi| = a^2 + b^2 = 1$. Let $\bar{z} = a - bi$
  denote the complex conjugate of $z$ and let $\text{Re}(z)$ denote the real
  part of a complex number. It is possible to show that $z$ is not a root of
  unity, i.e., $z^k \neq 1$ for all $k \in \mathbb{N}$. Let $\langle x_n
  \rangle$ be the sequence $x_n \triangleq (z^n + \bar{z}^n)/2 =
  \text{Re}(z^n)$. By direct calculation, this sequence has imaginary part zero
  and satisfies the recurrence
  \[
    x_{n + 2} = 2 a x_{n + 1} - (a^2 + b^2) x_n
  \]
  with $x_0 = 1$ and $x_1 = a$, so $\langle x_n \rangle$ is an order-2 rational
  LRS. By Theorem~\ref{thm:lrs-to-pa}, there exists a stochastic matrix $M$ and
  non-negative vectors $u, v$ such that
  \[
    x_n \geq 0 \iff u^T M^n v \geq 1/4
  \]
  for all $n$, where equality holds on the left if and only if equality holds on
  the right.  Note that $x_n = \text{Re} (z^n) \neq 0$ since $z$ is not a root
  of unity (so in particular $z^n \neq \pm i$), hence equality never holds on
  the right.  Letting $\langle y_n \rangle$ be the sequence $y_n = u^T M^n v$,
  the (unary) language with weights $\langle y_n \rangle$ is recognized by the
  DPA with input $u$, output $v$ and transition matrix $M$.
  Furthermore, the constant sequence $\langle 1/4 \rangle$ is recognizable by a
  DPA.

  Now we define a sequence $\langle w_n \rangle$ with $w_n = \max(y_n, 1/4)$.
  Since $\langle y_n \rangle$ and $\langle 1/4 \rangle$ are recognizable by
  DPAs, $\langle w_n \rangle$ is recognizable by an NPA whose initial state
  nondeterministically chooses between the two DPAs (see
  Remark~\ref{rem:initial}).
  Suppose for the sake of contradiction that it is also recognizable by a
  WFA. Then $\langle w_n \rangle$ is an LRS (by
  Proposition~\ref{prop:pa-to-lrs}) and hence so is $\langle t_n \rangle$ with
  $t_n = w_n - y_n$. If we now consider the zero set
  \begin{align*}
	S &
		= \{ n \in \mathbb{N} \mid t_n = 0 \} \\
	&
		= \{ n \in \mathbb{N} \mid y_n > 1/4 \} &
		&
		\text{($y_n \neq 1/4$)} \\
	&
		= \{ n \in \mathbb{N} \mid x_n > 0 \} &
		&
		\text{(Theorem~\ref{thm:lrs-to-pa})} \\
	&
		= \{ n \in \mathbb{N} \mid \text{Re}(z^n) > 0 \} &
		&
		\text{(by definition)},
  \end{align*}
  Theorem~\ref{thm:sml} implies that $S$ is the union of a finite set of indices
  and along with a finite number of arithmetic progressions. Note that $S$
  cannot be finite---in the last line, $z^n$ is dense in the unit circle since
  $z$ is not a root of unity---so there must be at least one arithmetic
  progression $\{ p + kn \mid n \in \mathbb{N} \}$. Letting $\langle r_n
  \rangle$ be
  \[
    r_n = (z^p \cdot (z^k)^n + \bar{z}^p \cdot (\bar{z}^k)^n)/2
    = \text{Re}(z^p \cdot (z^k)^n)
    = x_{p + kn} ,
  \]
  we have $p + kn \in S$, so $r_n > 0$ for all $n \in \mathbb{N}$, but this is
  impossible since it is dense in $[-1, 1]$ (because $z^k$ is not a root of
  unity for $k \neq 0$, so $z^p \cdot (z^k)^n$ is dense in the unit circle).

  Hence, the unary weighted language $\langle w_n \rangle$ can be recognized by
  an NPA but not by a WFA.
  \qed
\end{proof}

\section{Checking Language Equivalence of NPAs} \label{sec:equiv}

Given the coalgebraic NPA model, a natural question is whether there is a procedure to check language equivalence of NPAs.
We will show that language equivalence of NPAs is undecidable by reduction from the \emph{threshold problem} on DPAs.
Nevertheless, we can define a metric on the set of languages recognized by NPAs to measure their similarity.
While this metric cannot be computed exactly, it can be approximated to any given precision in finite time.

\subsection{Undecidability and Hardness}

\begin{theorem} \label{thm:equiv-hard}
	Equivalence of NPAs is undecidable when $|A| \ge 2$ and the $\Psc\D$-algebra on $[0, 1]$ extends the usual $\D$-algebra on $[0, 1]$.
\end{theorem}
\begin{proof}
	Let $X$ be a DPA and $\kappa \in [0, 1]$.
	We define NPAs $Y$ and $Z$ as follows:
	\begin{align*}
		Y = 
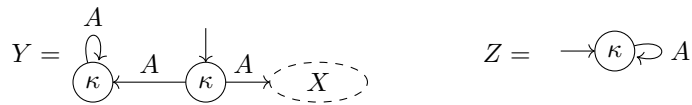
\begin{gathered}
			\begin{automaton}
				\node[initial above,state] (q0) {$\kappa$};
				\node[state] (q1) [left of=q0] {$\kappa$};
				\node[state,ellipse,dashed] (q2) [right of=q0] {\hspace{.3cm}$X$\hspace{.3cm}\mbox{}};
				\path[->]
				(q0) edge [swap] node {$A$} (q1)
				(q0) edge [pos=.4] node {$A$} (q2)
				(q1) edge [loop above] node {$A$} ();
			\end{automaton}
		\end{gathered} &
			&
			Z = \begin{gathered}
				\begin{automaton}
					\node[initial,state] (q0) {$\kappa$};
					\path[->]
					(q0) edge [loop right] node {$A$} ();
				\end{automaton}
			\end{gathered}
	\end{align*}
	Here the node labeled $X$ represents a copy of the automaton $X$---the transition into $X$ goes into the initial state of $X$.
	Note that the edges are labeled by $A$ to indicate a transition for every element of $A$.
	We see that $\lang_Y(\eword) = \kappa = \lang_Z(\eword)$ and (for $\alpha$ either $\minf$ or $\maxf$, as follows from Corollary~\ref{cor:minmax})
	\begin{align*}
		\lang_Y(av) &
			= (\alpha \circ \Conv)(\{\kappa, \lang_X(v)\}) &
			\lang_Z(av) &
			= \kappa.
	\end{align*}
	Thus, if $\alpha = \minf$, then $\lang_Y = \lang_Z$ if and only if $\lang_X(v) \ge \kappa$ for all $v \in A^*$; if $\alpha = \maxf$, then $\lang_Y = \lang_Z$ if and only if $\lang_X(v) \le \kappa$ for all $v \in A^*$.
	Both of these threshold problems were shown to be undecidable, for alphabets of size at least 2, by Blondel and Canterini~\cite[Theorem~2.1]{blondel2003}.
	\qed
\end{proof}

The situation for automata over unary alphabets is more subtle; in particular, the threshold
problem is not known to be undecidable in this case. However, there is a reduction to a
long-standing open problem on LRSs.

Given an LRS $\langle u_i \rangle$, the \emph{Positivity}
problem is to decide whether $u_i$ is non-negative for all $i \in \mathbb{N}$
(see, e.g.,~\cite{ouaknine2014positivity}). While the decidability of this
problem has remained open for more than 80 years, it is known that a decision
procedure for Positivity would necessarily entail breakthroughs in open problems
in number theory. That is, it would give an algorithm to compute the
\emph{homogeneous Diophantine approximation type} for a class of transcendental
numbers~\cite{ouaknine2014positivity}. Furthermore, the Positivity problem can
be reduced to the threshold problem on unary probabilistic automata.
Putting everything together, we have the following reduction.

\begin{corollary}
  The Positivity problem for linear recurrence sequences can be reduced to the
  equivalence problem of NPAs over a unary alphabet.
\end{corollary}
\begin{proof}
  The construction in Theorem~\ref{thm:equiv-hard} shows that the lesser-than
  threshold problem can be reduced to the equivalence problem for NPAs with
	$\maxf$ semantics, so we show that Positivity can be reduced to
  the lesser-than threshold problem on probabilistic automata with a unary alphabet. Given any rational LRS $\langle u_i \rangle$, clearly
  $\langle - u_i \rangle$ is an LRS as well, so by Theorem~\ref{thm:lrs-to-pa}
  there exists a rational stochastic matrix $M$ such that
  \[
    - u_n > 0 \iff u^T M^n v > 1/4
  \]
  for all $n$, where $u = (1, 0, \dots, 0)$ and $v = (0, 1, 0, \dots, 0)$.
  Taking $M$ to be the transition matrix, $v$ to be the input vector, and $u$ to
  be the output vector, the probabilistic automaton corresponding to the
  right-hand side is a nonsatisfying instance to the threshold problem with
  threshold $\leq 1/4$ if and only if the $\langle u_i \rangle$ is a satisfying
  instance of the Positivity problem.
  
  Applying Proposition~\ref{prop:relation} yields an analogous reduction from
  Positivity to the equivalence problem of NPAs with $\minf$ semantics.
\end{proof}

\subsection{Checking Approximate Equivalence}

The previous negative results show that deciding exact equivalence of
NPAs is computationally intractable (or at least very difficult, for a unary
alphabet). A natural question is whether we might be able to check approximate
equivalence. In this section, we show how to approximate a metric on
weighted languages. Our metric will be \emph{discounted}---differences in
weights of longer words will contribute less to the metric than differences in
weights of shorter words.

Given $c \in [0, 1)$ and two weighted languages $l_1, l_2 \colon A^* \to [0,
1]$, we define
\[
	d_c(l_1, l_2) = \sum_{u \in A^*}|l_1(u) - l_2(u)| \cdot \left(\frac{c}{|A|}\right)^{|u|} .
\]
Suppose that $l_1$ and $l_2$ are recognized by given NPAs. Since $d_c(l_1,
l_2) = 0$ if and only if the languages (and automata) are equivalent, we cannot
hope to compute the metric exactly. We can, however, compute the weight of any
finite word under $l_1$ and $l_2$. Combined with the discounting in the metric,
we can approximate this metric $d_c$ within any desired (nonzero) error.

\begin{theorem}
  There is a procedure that given $c \in [0, 1)$, $\kappa > 0$, and computable
  functions $l_1, l_2 \colon A^* \to [0, 1]$ outputs $x \in \R_+$ such that
  $|d_c(l_1, l_2) - x| \le \kappa$.
\end{theorem}
\begin{proof}
	Let $n = \lceil\log_c((1 - c) \cdot \kappa)\rceil \in \N$ and define
	\[
		x = \sum_{u \in A^*, |u| < n}|l_1(u) - l_2(u)| \cdot \left(\frac{c}{|A|}\right)^{|u|}.
	\]
  This sum is over a finite set of finite strings and the weights of $l_1(u)$ and
  $l_2(u)$ can all be computed exactly, so $x$ is computable as well. Now we can
  bound
	\begin{align*}
		|d_c(l_1, l_2) - x| &
      = \sum_{u \in A^*, |u| \ge n}|l_1(u) - l_2(u)| \cdot \left(\frac{c}{|A|}\right)^{|u|} \\
		&
			\le \sum_{u \in A^*, |u| \ge n}\left(\frac{c}{|A|}\right)^{|u|} \\
		&
			= \sum_{i \in \N, i \ge n}|A|^{i}\cdot \left(\frac{c}{|A|}\right)^{i} \\
		&
			= \sum_{i \in \N, i \ge n}c^{i} \\
		&
			= \frac{c^n}{1 - c} \le \kappa,
	\end{align*}
  where the last step is because $n \ge \log_c((1 - c) \cdot \kappa)$, and thus
  $c^n \le (1 - c) \cdot \kappa$, noting that $c \in [0, 1)$ and $\kappa > 0$.
	\qed
\end{proof}

We leave approximating other metrics on weighted languages---especially
nondiscounted metrics---as an intriguing open question.

\section{Conclusions} \label{sec:conclusion}
We have defined a novel probabilistic language semantics for nondeterministic
probabilistic automata (NPAs). We proved that NPAs are strictly more expressive
than deterministic probabilistic automata, and that exact equivalence is
undecidable. We have shown how to approximate the equivalence question to
arbitrary precision using a discounted metric.  There are two directions for
future work that we would like to explore. First, it would be interesting to see
if different metrics can be defined on probabilistic languages and what
approximate equivalence procedures they give rise to. Second, we would like to
explore whether we can extend logical characterization results in the style of
Panangaden et
al.~\cite{DBLP:conf/icalp/FijalkowKP17,DBLP:conf/lics/DesharnaisEP98}.  Finally,
it would be interesting to investigate the class of languages recognizable by
our NPAs.

\paragraph{Related Work.}
There are many papers studying probabilitic automata and variants thereof. The
work in our paper is closest to the work of Segala~\cite{segala-thesis} in that
our automaton model has both nondeterminism and probabilistic choice. However,
we enrich the states with an output weight that is used in the definition of the
language semantics. Our language semantics is coarser than probabilistic
(convex) bisimilarity~\cite{segala-thesis} and bisimilarity on
distributions~\cite{DBLP:conf/concur/HermannsKK14}. In fact, in contrast to the
hardness and undecidability results we proved for probabilistic language
equivalence, bisimilarity on distributions can be shown to be
decidable~\cite{DBLP:conf/concur/HermannsKK14} with the help of convexity. The
techniques we use in defining the semantics are closely related to the recent
categorical understanding of bisimilarity on
distributions~\cite{DBLP:conf/concur/Bonchi0S17}.

\bibliographystyle{plainurl}
\bibliography{header,main}

\end{document}